\documentclass[proceedings,,copyright,creativecommons]{eptcs}
\providecommand{\event}{TARK 2015} 
\usepackage{graphicx}
\usepackage{breakurl}             
\usepackage{amssymb}
\usepackage{amsthm}
\title{Relating Knowledge and Coordinated Action:\\
The Knowledge of Preconditions Principle}
\author{Yoram Moses\footnote{The Israel Pollak academic chair at Technion. This work was supported in part by ISF grant 1520/11.}
\institute{Technion---Israel Institute of Technology}
\email{moses@ee.technion.ac.il}}
\def\titlerunning{The Knowledge of Preconditions Principle} 
\def\authorrunning{Yoram Moses}
\usepackage{amsmath}
\usepackage[capitalise]{cleveref}
\usepackage{fancybox}
\usepackage[framemethod=TikZ]{mdframed}
\newcommand{\eqdef}{\triangleq}

\newcommand{\recgo}{\psi_{\mathsf{go}}}

\newcommand{\fire}{\mathsf{fire}}
\newcommand{\go}{\mathsf{go}}

\newcommand{\KoP}{\mbox{{\bf K}$\!\!\!\!\;\;${\it o}{\bf P}}\/}
\newcommand{\hqed}{\hfill\ensuremath{\square}}

\newcommand{\sfa}{\alpha}

\newcommand{\Agents}{{\mathbf{A}}}
\newcommand{\Proc}{{\mathbb{P}}}

\newcommand{\Lang}{\ensuremath{{\cal L}}}
\newcommand{\LKn}{\Lang_n^K}

\newcommand{\sat}{\models}

\newcommand{\imp}{\Rightarrow}

\newcommand{\Max}{\mathsf{Max}}

\newcommand{\CtMb}{Computing the Max}
\newcommand{\CtM}{\mbox{{\sc CtM}}}

\newcommand{\defemph}[1]{\textbf{\textit{#1}}}
\newcommand{\bm}[1]{\mathbf{#1}}
\newtheorem{theorem}{Theorem}[section]

\newtheorem{corollary}[theorem]{Corollary}

\newtheorem{definition}[theorem]{Definition}
\newtheorem{claim}{Claim}
\newtheorem{example}{Example}
\newtheorem{observation}{Observation}

\newcommand{\Nat}{{\mathbb{N}}}

\makeindex

\newcommand{\True}{T}
\newcommand{\False}{F}

\newcommand{\Pts}{\mathsf{Pts}}
\newcommand{\ddo}{\mathtt{does}}
\newcommand{\dDid}{\mathtt{did}}

\newcommand{\rlsim}[1]{\thickapprox_{#1}}

\newcommand{\G}{{\cal G}}

\renewcommand{\bm}[1]{#1}%
\begin{document}
\maketitle 
\begin{abstract}
The \defemph{Knowledge of Preconditions} principle (\KoP) is proposed as a 
widely applicable connection between knowledge and action in multi-agent systems. Roughly speaking, it asserts that if some condition~$\varphi$ is a necessary condition for performing a given action $\sfa$, then {\em knowing}~$\varphi$ is also a necessary condition for performing~$\sfa$. 
Since the specifications of tasks often involve necessary conditions for actions, the \KoP\ principle shows that such specifications induce knowledge preconditions for the actions. Distributed protocols or multi-agent plans that satisfy the specifications  must ensure that this knowledge be attained, and that it is  detected by the agents as a condition for action. 
The knowledge of preconditions principle is formalised in the runs and systems framework, and is proven to hold in a wide class of settings. 
Well-known connections between knowledge and coordinated action are extended and shown to derive directly from the \KoP\ principle: a {\em common knowledge of preconditions} principle is established showing that common knowledge is a necessary condition for performing simultaneous actions, and a  {\em nested knowledge of preconditions} principle is proven, showing  that coordinating actions to be  performed in linear temporal order requires a corresponding \mbox{form of nested knowledge}. 
\end{abstract}

{\bf Keywords:~}{Knowledge, multi-agent systems, common knowledge, nested knowledge, coordinated action, knowledge of preconditions principle.}

\section{Introduction}
While epistemology, the study of knowledge, has been a topic of interest in philosophical circles for centuries and perhaps even millennia, in the last half century it has seen a flurry of activity and applications in other fields such as AI \cite{MH}, game theory \cite{Au} and distributed computing \cite{HM1}.  
 At least in the latter two fields a particular, information-based, notion of knowledge plays a prominent and useful role. 
 
This paper proposes an essential connection between knowledge and action in such a setting. 
Using $\ddo_i(\alpha)$ to denote ``{\it Agent~$i$ is performing action~$\alpha$}'' and $K_i\varphi$ to denote that Agent~$i$ knows the fact~$\varphi$, the connection can  intuitively be formulated as follows: 
\vspace{4pt}
\begin{mdframed}[roundcorner=4pt]
 \begin{center}
\underline{The  \defemph{\sc{Knowledge of Preconditions}} Principle (\KoP):}
\end{center}
 \vspace{3.5mm}
\begin{tabular}{lrl}
$~$\hspace{2.7cm}If  & {$\bm\varphi$} & is a necessary condition for~~$\ddo_i(\sfa)$ \\[.91ex]
$~$\hspace{2.7cm}then & ${\bm{K_i\varphi}}$ & 
is a necessary condition for~~$\ddo_i(\sfa)$
\end{tabular}
 \vspace{1.4mm}
\end{mdframed}

This statement appears deceptively simple. 
In fact, many successful applications of knowledge to the design and analysis of distributed protocols over the last three decades are rooted in the \KoP. Moreover,  some of the deeper insights obtained by knowledge theory in this field can be derived in a fairly direct fashion from the \KoP. We will argue and demonstrate  that this principle lies at the heart of coordination in many distributed and multi-agent systems. 

This paper is structured as follows. Section~\ref{sec:case} illustrates the central role of knowledge in a natural distributed systems application. Section~\ref{sec:kop} provides a high-level discussion of  the knowledge of preconditions principle and its connection to coordinating actions. In \cref{sec:runs} we review and discuss the modelling of knowledge in the runs and systems model of distributed systems based on~\cite{FHMV}.  A formal statement and proof of the \KoP\ are presented in \cref{sec:formal-kop}. Then, in \cref{sec:ck}, the \KoP\ is used to establish a {\em common knowledge of preconditions} principle. It states that in order to perform simultaneously coordinated actions, agents must first attain common knowledge of any of the actions' preconditions.  An example of its use is provided in \cref{sec:firing}. \cref{sec:nested} present an additional use of the \KoP, and shows that coordinating a sequence of actions to occur in a prescribed temporal order requires attaining nested knowledge of their preconditions. 
Finally, \cref{sec:discussion} discusses additional applications, extensions  and future directions. 


\subsection{The Case for Knowledge in Distributed Systems}
\label{sec:case}
Why should  knowledge
play a central role in distributed computing? 
As pointed out in~\cite{HM1}, most everyone who designs or even just tries to study the workings of a distributed protocol is quickly found talking in terms of knowledge, making statements such as {\em ``once the process receives an acknowledgement, it knows that the other process is ready\ldots''}.
An essential aspect of distributed systems is the fact that an agent chooses
 which action to perform based on the  local information available to it, which typically provides only  a partial view of the overall state of the system. 
To get a sense of the role of knowledge in distributed systems, consider the following example. 

\begin{example}{$~\!\!\!$} 
\label{ex:Max}
\begin{figure}
\label{fig:1.1}
\centering
\includegraphics[width=4.4in,height=2.2in]{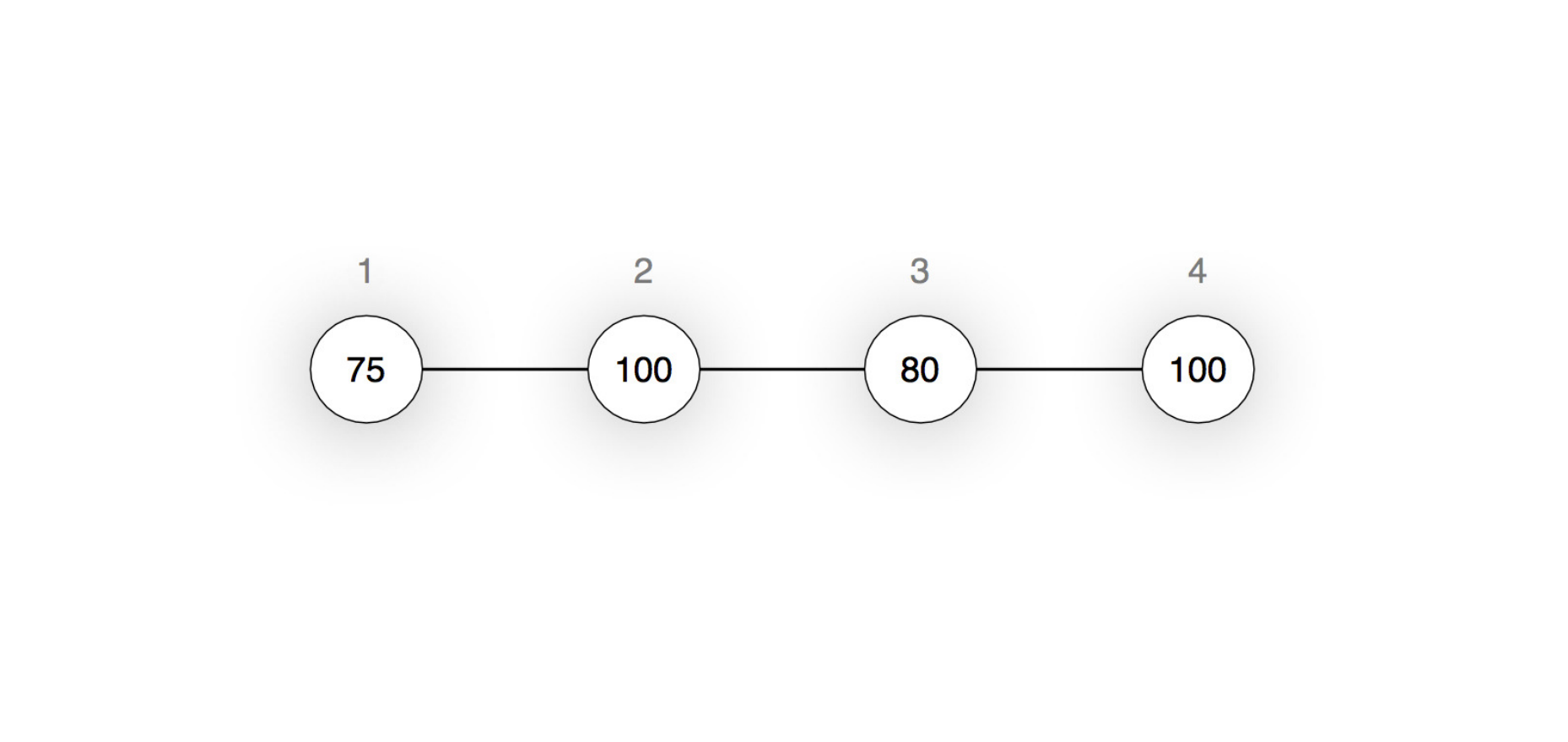}
\vspace{-12mm}
\caption{A simple four-agent system}
\end{figure}
\noindent
Given is a distributed network modeled 
by a graph, with agents  located at the nodes, and the edges standing for 
communication channels (see Figure~\ref{fig:1.1}). In the problem we shall call \defemph{\CtMb} (or $\CtM$ for short), each agent~$i$ starts out with a natural number {$v_i\in\Nat$} as an initial value. The goal is to have Agent~1 print 
the maximum of all of the initial values (we denote this value by~$\Max$), and print nothing else. 
In the instance depicted in Figure~\ref{fig:1.1}, the maximal value happens to be~100. 
Initially, Agent~1 clearly can't print its own initial value of 75. 
Suppose that  Agent~1 receives a message~$\mu\eqdef\,\,$``$v_2=100$''  from Agent~2 reporting that its value is~$100$. 
At this point, Agent~1 has access to the maximum, and printing~100 would satisfy the problem specification. Compare this with a setting that is the same in all respects, except that Agent~3's value is $v_3=150$. 
In this case, of course, $\Max\ne 100$ and so printing~100 is forbidden. 
But if Agent~1 can receive the same message~$\mu$ under similar circumstances in both scenarios, 
then it is unable to distinguish whether or not $\Max=100$ upon receiving~$\mu$. 
Intuitively, even in the first scenario, the agent does not {\em know} that $\Max=100$.

What information does Agent 1 need, then, in order to be able to print the maximum? 
Notice that it is not necessary, in general,  to collect {\em all} of the initial values in order to print the maximum. 
For example, suppose that the agents follow a bottom-up protocol in which 
values are sent from right to left, starting from Agent~4, and every agent passes to the left the larger of its own value and the value it received from its neighbor on the right (if such a neighbor exists). 
In this protocol, Agent~1 can clearly print the maximum after receiving 
the message {\it ``$v_2=100$''}, and seeing just one value besides its own. 
Interestingly, even collecting all of the values is  {\em not  a sufficient condition} for printing the~$\Max$. Imagine a setting in which the network is as in Figure~\ref{fig:1.1}, but Agent~1 considers it possible that there are more than four nodes in the network. In this case, even if Agent~1 receives all (four) values, it may still need to wait for proof that there is no additional, larger, value in the system. 
\hqed
\end{example}

\CtM\ is a simplified example in the spirit of many distributed systems applications. In fact, a central problem called {\it Leader Election}, for example, is often solved by computing a node with maximal ID \cite{AWbook,lelann}.
The solution to such a problem is typically in the form of a set of short computer programs (jointly constituting a {\em distributed protocol}), each executed at one of the nodes. 
When the nodes follow such a protocol, the resulting execution should satisfy the problem specification.  Of course, the programs are written in a standard programming language, without any reference to knowledge or possibility. 
In the vast majority of cases, the programs in question do not enumerate and/or explore possible states or scenarios. Indeed, the program designer is typically unfamiliar with formal notions of knowledge. 
This being the case, what sense does it make to talk of Agent~1 in Example~\ref{ex:Max} ``knowing'' or ``not knowing'' that $\Max=c$? Can it make sense to say that the Agent ``considers it possible that there may be more than four nodes in the system''? After all, we may be talking about a 10-line program. It has no soul. Does it have thoughts, doubts and mental states?

Since agents act based on their local information, a protocol designer must ensure that agents obtain the necessary information for a given task, and that this information is applied correctly. 
Using the information-based notion of knowledge, the designer can ascribe knowledge to an agent without requiring it to have a soul, feelings, and self-awareness. 
As seen in the \CtM\ example, it is natural to think  in terms of whether or not Agent~1 knows $\Max=c$ at any given point in a run of a protocol solving~\CtM. (A formal definition of knowledge will be provided in Section~\ref{sec:runs}.)
Suppose that a protocol is designed to solve  \CtM\ in networks that may have a variety of sizes. If Agent~1 does not start out with local information ensuring that there are no more than four nodes in the system, then from the point of view of an outside observer the agent can be thought of as  ``considering it possible'' that there may be more than four nodes.

Even in a simple network as in Figure~\ref{fig:1.1}, the \CtM\ problem can be posed in different models, which can differ in essential aspects. 
A solution to $\CtM$ in one model might not solve the problem in another model. Indeed,  the  rationale behind distinct solutions, as well as their details, may vary considerably. 
Are there common features shared by all solutions to \CtM?

Interestingly, all solutions to $\CtM$, in all models, share one property: Agent~1 must \defemph{know} that $\bm{\Max=c}\,$ in order to print the value~$c$. 
Indeed, the ability to print the answer in a protocol for~$\CtM$ reduces to detecting when the $\Max$ value is known. 
Of course, once Agent~1 knows that $\Max=c$ it can safely print~$c$. Hence, knowing that $\Max=c$ is not just necessary, but also a sufficient condition for printing~$c$.
The \CtM\ problem shows that knowledge and attaining knowledge can be a central and crucial aspect of a standard distributed application. 

The need to know $\Max=c$ in solving $\CtM$ suggests that we consider a  natural question: {\em When does Agent~1 know that $\Max=c$?} The answer is less straightforward than we might initially expect. What is known depends in a crucial way on the protocol that the agents are following. 
Thus, in the setting of \cref{fig:1.1}, if the agents follow the bottom-up protocol, then Agent~1 knows the maximum once  it receives a single message from Agent~2.
Knowledge is also significantly affected by features of the model. In $\CtM$, if there is an upper bound (say 100) on the possible initial values, then an agent that sees this value knows the maximum. Knowledge about the network topology and properties of communication play a role as well. 
For example, consider a model in which Agent~1 has a clock, and a single clock cycle suffices for a message to be delivered, digested, and acted upon. 
Suppose that the protocol is such that all agents start simultaneously at time~0 and 
 an agent forwards a value towards Agent~1 only if this value is larger than any value 
 it has previously sent. Then in the network of Figure~\ref{fig:1.1} Agent~1 will receive a message with value $100$ from Agent~2 at time~1, and no further messages. If Agent~1 knows that the diameter of the network is~3, it will not know the maximum upon receiving this message. However, without receiving any further messages,  at time~3 Agent~1 {\em will} know that the maximum is 100; no larger value can be lurking in the system.

\subsection{K{o}P and Coordination} 
\label{sec:kop}
The fact that  $\Max=c$ is a necessary condition for printing~$c$ is an essential feature of the \CtM\ problem. We have argued that, in fact, $K_1(\Max=c)$ is also a necessary condition for printing~$c$, as the \KoP\ would suggest. But this is just one instance. Let us briefly consider another example. 
\begin{example}
Consider a bank whose ATMs are designed in such a way that an ATM will dispense cash only to a customer whose account shows a sufficiently large positive balance.
 Along comes Alice, who has a large positive balance,  and tries to obtain a modest sum from the ATM. On this day, however, the ATM is unable to communicate with the rest of the bank and it declines to pay Alice. Thus, despite the fact that Alice has good credit, the ATM frustrates her and denies her request. Apparently, given its specification, the ATM has no choice. Intuitively, in order to satisfy the credit restriction, the ATM needs to {\em know} that a customer has good credit before dispensing cash. If the ATM may pay a customer that is not known to have good credit, there will be possible scenarios in which the ATM will violate its specification, and pay a customer that does not have credit. Notice, however, that the specification said nothing about the ATM's knowledge. It only imposed a restriction on the ATM's action, based on the state of Alice's account. 
 \hqed
\end{example}
Both the \CtM\ problem and the ATM example are instances in which the \KoP\ clearly applies. 
The intuitive argument for why the \KoP\ should apply very broadly is straightforward. 
If $\varphi$ is a necessary condition for performing~$\sfa$, and agent~$i$ ever performs~$\sfa$ without knowing~$\varphi$, then there should be a possible scenario that is indistinguishable to agent~$i$, in which~$\varphi$ does not hold. Since the two scenarios are indistinguishable, the agent can perform~$\sfa$ in the second scenario, and violate the requirement that~$\varphi$ is a necessary condition. A formal statement and proof  requires a definition of necessary conditions, knowledge, as well as capturing a sense in which an action at one point implies the same action at any other, indistinguishable,  point. This will be done in \cref{sec:formal-kop}. 

Most tasks in distributed systems are described by way of a specification. Such specifications typically impose a variety of necessary conditions for actions. The \KoP\ implies that even though such specifications often do not explicitly discuss  the agents' knowledge, they do in fact impose knowledge preconditions. Observe that the \KoP\ applies to a task regardless of the means that are used to implement it. Any engineer implementing a particular task will have to ensure that preconditions are known when actions are taken. This is true whether or not the engineer reasons 
explicitly in terms of knowledge, and it is true even if the engineer is not even aware of the knowledge terminology. (Normally, neither may be the case, of course.) The need to satisfy the \KoP\ suggests that the design of distributed implementations must involve at least two steps. One is to make sure that the required knowledge is made available to an agent who needs to performed a prescribed action, and the other is ensuring that the agent detect that it knows the required preconditions. This is quite different from common practice in engineering distributed implementations~\cite{schneider1990}. 

We remark that the \KoP\ can be expected to hold in a variety of multi-agent settings well beyond the realm of distributed systems. Thus, for example, suppose that a jellyfish is naturally designed so that it will never sting its own flesh. By the \KoP, the cell activating the sting at a given point needs to {\em know} that it is not stinging the jellyfish's body when it ``fires'' its sting. The jellyfish is thus designed with some form of a ``friend or foe'' mechanism that is used in the course of activating the sting. 
Various biological activities can similarly be considered in light of the \KoP: How does the organism know that certain preconditions are met? Our last example will come from the social science arena.  Suppose that a society designs a legal system, that is required to satisfy the constraint that only people who are guilty of a particular crime are ever put in jail for committing this crime. By the \KoP, the judge (or jury) must {\em know} that the person committed this crime in order to send him to jail.

As discussed above, specifications impose preconditions. Typically, 
these conditions relate an action to facts about the world (e.g., the maximal value,  or the customer's good credit). In many cases, however, actions of different agents need to be coordinated.  
Consider a variant of \CtM\ in which in addition to Agent~1 printing the maximum, Agent~4 needs to perform an action (say print the same value or print the minimal value), but not before Agent~1 does. Then Agent~1 performing her action is a condition for~4's action. In particular, Agent~4 would need to know that Agent~1 has already come to know $\Max=c$ for some~$c$ before~4 acts. In some cases, the identity of actions performed needs to be coordinated. 

For a final example, suppose that Alice should perform an action~$\sfa_A$ only if~Bob performs an action~$\sfa_B$ at least~5 time steps earlier. Then she needs to know that Bob acted at least~5 steps before when she acts. Indeed, if~$\psi$ is a necessary condition for~$\sfa_B$, then Alice must know that 
``Bob knew~$\psi$ at least~5 time steps ago'' when she acts, since knowing~$\psi$ is a necessary condition for Bob's performing $\alpha_B$ (see~\cite{BzMTARK2013,BZMacm}).  
As these examples illustrate, given \KoP, coordination can give rise to nested knowledge. 

Simple instances of the \KoP\ are often quite straightforward. Ensuring and detecting $K_1(Max=c)$ is often fairly intuitive, and it not justify the overhead involved in developing a theory of knowledge for multi-agent systems. However, satisfying statements involving nested knowledge in particular models of computation can quickly become nontrivial. For this, it is best to have a clear mathematical model of knowledge in multi-agent systems. The next section reviews the runs and systems model.

\section{Modeling Knowledge Using Runs and Systems}
\label{sec:runs}
We now review the runs and systems model of knowledge of \cite{FHMV,HM1}. The interested reader should consult \cite{FHMV} for more details. A \defemph{global state} is an ``instantaneous snapshot'' of the system at a given time. 
Let~$\G$ denote a set of global states. 
Time will be identified with the natural numbers~$\Nat=\{0,1,2,\ldots\}$ for ease of exposition. 
A \defemph{run} is a function \mbox{\,\defemph{r}\hspace{.2mm}$:\,\Nat\to\G\,$} associating a global state with each instant of time. 
Thus, $r(0)$ is the run's initial state, $r(1)$ is the next global state, and so on. 
A \defemph{system} is a set~$\bm{R}$ of runs. 
The same global state can appear in different runs, and in some systems may even appear more than once in the same run. 

A central notion in our framework is that of an agent's {\em local state}, whose role is to capture the agent's local information at a given point.  The precise details of the local state depend on the application. It could be the complete contents of an agent's memory at the given instant, or the complete sequence of events that it has observed so far. for example. 
 The rule of thumb is that the local state should consist of the local information that the agent may use when deciding which actions to take. 
Thus, for example, if agents are finite-state machines, it is often natural to identify an agent's local state with the automaton state that it is in. 
Formally, we assume that 
every global state determines a unique \mbox{\defemph{local state}} for each agent.
We denote agent~$i$'s local state in the global state $r(t)$ by $r_i(t)$.
Moreover, a global state with~$n$ agents $\Agents=\{1,\ldots,n\}$ will have the form $r(t)=\langle r_e(t),r_1(t),\ldots,r_n(t)\rangle$, where $r_e(t)$ is called the local state of the {\em environment}, and will serve to represent all aspects of the global state that are not included in the agents' local states. For example, it could represent messages in transit, the current topology of the network including what links may be down, etc.

\subsection{Syntax and Semantics}
\label{sec:syntax}
We are interested in a propositional logic of knowledge, in which propositional facts and epistemic facts can be expressed. Facts will be considered to be true or false at a point $(r,t)$, with respect to a system~$R$. 
More formally, given a set~$\Phi$ of primitive propositions and a set  $\Proc=\{1,\ldots,n\}$ of the agents in the system, we define a propositional language $\LKn(\Phi)$ by closing~$\Phi$ under negation~`$\neg$' and conjunction~`$\wedge$', as well as under knowledge operators~$K_i$ for all $i\in\Proc$ (see~\cite{HM2}). Thus, for example, if $p,q\in\Phi$ are primitive propositions and $i,j\in\Proc$ are agents, then~ $\neg K_ip\wedge K_jK_i\neg K_jq$ ~is a formula in $\LKn(\Phi)$.
We typically omit the set~$\Phi$ and call $\LKn$ the language for knowledge with~$n$ agents. 

In a multi-agent system facts about the world, as well as the knowledge that agents have, can change dynamically from one time point to the next. We thus consider the truth of formulas of $\LKn$ at {\em points} of a system~$R$, where a point is a pair $\bm{(r,t)}\in R\times\Nat$, and it is used to refer to time~$t$ in the run~$r$. We denote the set of points of a system~$R$ by $\Pts(R)\eqdef R\times\Nat$. Points will play the role of states of a Kripke structure. 

The set~$\Phi$ of primitive propositions used in the analysis of any given multi-agent system~$R$ will depend on the application. Their truth at the points of the system needs to be explicitly defined. This is done by an {\em interpretation} $\pi:\Phi\times\Pts(R)\to\{\True,\False\}$, where $\pi\big(q,(r,t)\big)=\True$ means that the proposition~$q$ holds at $(r,t)$. 
Formally, an {\em interpreted system} w.r.t.\ a set~$\Phi$ of primitive propositions is a pair $(R,\pi)$ 
consisting of the system~$R$ and interpretation~$\pi$ for~$\Phi$ over $\Pts(R)$. Just as we typically omit explicit reference to~$\Phi$, we shall omit $\pi$ as well, when this is unambiguous.

We assume from here on that  the environment's state~\mbox{$r_e(t)$} in a global state $r(t)$ contains a ``{\em history}'' component~$h$ that records  all actions taken by all agents at times $0$,$1$,\ldots,$t-1$. Formally, we take~$h$ to be a set of triples $\langle{\sfa,i,t'}\rangle$, which grows monotonically in time. An action $\sfa$ is considered to be performed by~$i$ at the point $(r,t)$ if and only if the triple $\langle{\sfa,i,t}\rangle$, denoting that action~$\sfa$ was performed by agent~$i$ at time~$t$,  appears in the history component~$h$ of $r_e(t')$ for all times $t'>t$.%
 \footnote{Our definition does not imply or assume that the actions are observed, observable or recorded by any of the agents. Whether that is the case depends on the application.}    
For the analysis in this paper, we will also assume that~$\Phi$ includes propositions of the form 
$\ddo_i(\sfa)$ and $\dDid_i(\sfa)$ for agents $i\in\Proc$ and actions~$\alpha$. 
  With this assumption, what actions  are performed at any given point $(r,t)$ is uniquely determined by the run~$r$. 

We will consider interpretations $\pi$ that, on these propositions, are defined by 
\vspace{1.5mm}

\noindent\begin{tabular}{l l l}
$\pi\big(\ddo_i(\sfa),(r,t)\big)=T$& iff & agent~$i$ performs $\alpha$ at $(r,t)$\\[2ex] 
$\pi\big(\dDid_i(\sfa),(r,t)\big)=T$ & iff & $\pi\big(\ddo_i(\sfa),(r,t')\big)=T$ ~ 
 holds for some $t'\le t$
\end{tabular}
\vspace{3mm}

We allow $t'=t$ in the definition of $\dDid_i(\sfa)$  for technical convenience; it simplifies our later analysis slightly. 
 

Our model of knowledge will follow the standard Kripke-style possible worlds approach. The possibility relations that we use are induced directly from the system~$R$ being analyzed; two points are considered indistinguishable to an agent if its local states at the two points are the same. 
 More formally: 
\begin{definition}
\label{def:ind}
If $r_i(t)=r'_i(t')$, then $(r,t)$ and $(r',t')$ are called \defemph{indistinguishable to~$\bm{i}$}, denoted by 
\mbox{$(r,t)\bm{\rlsim{i}}(r',t')$}. 
%
\end{definition}


Formulae of $\LKn$ are interpreted at a point $(r, t)$ of an
interpreted system $(R, \pi)$ by means of the satisfaction
relation `$\models$', which is defined inductively by:
\begin{itemize}
\item[] $\!\!\!\!\!\!(R,r,t) \models p$ iff $(r,t) \in \pi(p)$;
\item[] $\!\!\!\!\!\!(R,r,t) \models \neg\varphi$ iff $(R,r,t) \not\models
  \varphi$;
\item[] $\!\!\!\!\!\!(R,r,t) \models \varphi \wedge \psi$ iff both $(R,r,t)
  \models\varphi$ and \mbox{$(R,r,t) \models\psi$};
\item[] $\!\!\!\!\!\!(R,r,t)\models K_i\varphi$ iff $(R,r',t') \models
  \varphi$  for all $(r',t')\in\Pts(R)$
  ~such that $(r',t')\bm{\rlsim{i}}(r,t)$.~ 
\end{itemize}

 We say that~$\bm\varphi$ is \defemph{valid in} the system~$\bm{R}$, and write  $\bm{R\sat\varphi}$,  if $(R,r,t)\sat\varphi$ for all points $(r,t)\in\Pts(R)$. We say that $\bm\varphi$~\defemph{validly implies $\bm\psi$ in}~$\bm{R}$ if $\varphi\imp\psi$ is valid in~$R$.
Since, by Definition~\ref{def:ind},  the~$\rlsim{i}$ relations are equivalence relations, each knowledge operator $K_i$ satisfies the S5 axiom system~\cite{HM2}. In particular, it satisfies the \defemph{knowledge property} (or Knowledge Axiom) that $K_i\varphi\imp\varphi$ is valid in all systems. 

It is instructive to relate our modeling using runs and systems to standard multi-agent Kripke structures.  
For every system~$R$ there is a corresponding Kripke structure $M_R=(S_R,\pi,\sim_1,\ldots,\sim_n)$ for~$n$ agents such that $S_R=\Pts(R)$ and $\mbox{`$\sim_i$'}=\mbox{`$\rlsim{i}$'}$ for every~$i$.
They correspond in that 
$(R,r,t)\sat\varphi\,$ iff $\,M_R,(r,t)\sat\varphi$ is guaranteed for all $(r,t)\in\Pts(R)=S_R$ and $\varphi\in\LKn(\Phi)$ (for details, see~\cite{FHMV}). 

The system~$R$ will determine the space of possible runs and possible points, which play a crucial role in determining the truth of facts involving knowledge. 
For example, consider a run~$r$ in which  Alice sends Bob a message at time~1, and Bob receives it at time~2.  If~$R$ is a system in which messages may be lost, or may take longer than one time step to be delivered, then 
 Alice would not know at time~2 \big(i.e., w.r.t.~$(R,r,2)$\big)  that her message has been delivered, because there is another run $r'\in R$ 
 that she cannot tell apart from~$r$ at time~2, in which her message is not (or not yet) delivered by that time.  The same run~$r$ also belongs to another  system~$R'$  in which messages are always reliably delivered in exactly one round. 
With respect to $(R',r,2)$, however, Alice {\em would} know at time~2 that her message has been delivered.

Our definition of knowledge is rather flexible and widely applicable. The set~$R$ of the possible runs immediately induces what the agents know. 
Observe that the definition of knowledge is completely external. It ascribes knowledge to agents in the system even if the protocol they follow, as well as the actions that they perform, do not involve the knowledge terminology in any way.  Moreover, the agents do not need to be complex or sophisticated for the definition to apply.  Indeed, in a model of a very simple system consisting of a bed lamp and its electric cable, a switch in the OFF state can be said to know that the lamp is not lit;  what the same switch would know in the ON state would depend on the system~$R$ under consideration, which determines the runs considered possible. E.g.,  if~$R$ contains a run in which the lamp is burnt out, then in the ON state the switch would not know that the lamp is shining light. On the other hand, if the lamp can never burn out, and the cord, plug and switch are in proper working order in all runs of~$R$, then  in the ON state the switch {\em would} know that the lamp is shining light. 
As this example shows, knowledge under this definition does not require the ``knower'' to compute what it knows. Indeed, this definition of knowledge is not sensitive to the computational complexity of determining what is known. 
In most cases, of course, we will ascribe knowledge to agents or components that can perform actions, which is not the case in the light switch example. And agents might need to explicitly establish whether they know relevant facts. 
We now provide a statement and proof of the knowledge of preconditions principle~\KoP.

\section{Formalizing the Knowledge of Preconditions Principle}
\label{sec:formal-kop}

Intuitively, the \KoP\  states that if a particular fact~$\psi$ is a necessary condition for an agent to perform an action~$\sfa$, then the agent must in fact {\em know} $\psi$ in order to act. In other words, {\em knowing}~$\psi$ is also a necessary condition for performing the action. 
We formalize the claim and prove it as follows. 
We say that~$\bm{\psi}$  \defemph{is a necessary condition for
$\bm{\ddo}_i(\sfa)$  in}~$\bm{R}$ if $(R,r,t)\sat\ddo_i(\sfa)$ 
holds only if $(R,r,t)\sat{\psi}$, for all $(r,t)\in\Pts(R)$.
 Clearly,  
 the customer's good credit is a necessary condition for the ATM dispensing cash. That is, suppose that a bank makes use of a correct implementation of an ATM protocol, which satisfies the credit requirement. Then, in the system~$R$ consisting of the set of all possible histories (runs) of the bank's (and the ATM's) transactions, good credit is a necessary condition for receiving cash from the ATM.

It is often of interest to consider facts whose truth depends only on a given agent's loca state. 
Such, for example, may be the receipt of a message, or the observation of a signal,  by the agent. 
Whether $x=0$ for a local variable~$x$, for example, would be a natural local fact. Moreover, if an agent has perfect recall, then any events that it has observed in the past will give rise to local facts. Finally, since knowledge is defined based on an agent's local state, then a fact of the form $K_i\varphi$ constitutes a local fact. Indeed, there is a simple way to define the local facts above using knowledge. Namely, we say that~$\bm\varphi$ \defemph{is i-local in}~$\bm{R}$ if $\,R\sat(\varphi\imp K_i\varphi)$. 

\vspace{1mm}
The formalism of \cite{FHMV} defines protocols as explicit objects, and defines {\em contexts} that describe the possible initial states and the model of computation. 
This provides a convenient and modular way of constructing systems. Namely, given a protocol~$P$ and a {\em context}~$\gamma$, the system~$R=R(P,\gamma)$ is defined to be the set of all runs of protocol~$P$ in~$\gamma$. The runs of this system embody all of the properties of the context, as they arise in runs of~$P$. This includes, for example, any timing assumptions, possible values encountered, possible topologies of the network, etc. They also embody the relevant properties of the protocol, because in all runs considered possible the agents follow~$P$. 

In this paper, we do not define protocols and contexts. Rather, we treat the \KoP\ in a slightly simpler and more abstract setting. 
We say that an action~$\sfa$ is a  \defemph{conscious action for~$\bm{i}$ in~$\bm{R}$} if 
$i$'s local state completely determines whether~$i$ performs~$\sfa$. If its local state at two points $(r,t)$ and $(r',t')$ of~$R$ is the same, then $(R,r,t)\sat\ddo_i(\sfa)$ iff $(R,r',t')\sat\ddo_i(\sfa)$. 
Conscious actions are quite prevalent in many systems of interest. 
For example, suppose that agent~$i$ follows a deterministic protocol, so that its action at any given point is a function of its local state. If, in addition, agent~$i$ is allowed to move at every time step, then all of its actions are conscious actions. 
We remark that, since conscious actions depend on an agent's local state, then if~$\sfa$ is conscious for~$i$ in~$R$ then 
$(R,r,t)\sat \ddo_i(\sfa)$ holds iff $(R,r,t)\sat K_i\ddo_i(\sfa)$ does, for all $(r,t)\in\Pts(R)$. 

We are now ready to prove a formal version of the \KoP: 

\begin{theorem}[The \KoP\ Theorem]
\label{thm:kop-1}
Let $\sfa$ be a conscious action for~$i$ in~$R$. 
If   {$\psi$}  is a necessary condition for {$\ddo_i(\sfa)$} in~$R$,  
then  ${K_i\psi}$  is also a necessary condition for ${\ddo_i(\sfa)}$ in~$R$.
\end{theorem}
\begin{proof}
We will show the contrapositive. 
Let $\sfa$ be a conscious action for~$i$ in~$R$, and assume that $K_i\psi$ is \underline{not} a necessary condition for $\ddo_i(\sfa)$ in~$R$. 
Namely, there exists a point $(r,t)\in\Pts(R)$ such that 
both  $(R,r,t)\sat\ddo_i(\sfa)$ and $(R,r,t)\not\sat K_i\psi$. 
Given the latter, we have by the definition of `$\sat$' for $K_i$ that there exists a point $(r',t')\in\Pts(R)$ such that 
both $(r',t')\rlsim{i}(r,t)$ and $(R,r',t')\not\sat\psi$. 
Since~$\sfa$ is a conscious action for~$i$ in~$R$ and $(R,r,t)\sat\ddo_i(\sfa)$ we have that $(R,r,t)\sat K_i\ddo_i(\sfa)$. 
It follows from $(r',t')\rlsim{i}(r,t)$ by the definition of `$\sat$' for~$K_i$ that $(R,r',t')\sat\ddo_i(\sfa)$ holds. But since $(R,r',t')\not\sat\psi$, we conclude  that 
$\psi$ is \underline{not} a necessary condition for~$\ddo_i(\sfa)$ in~$R$, establishing the countrapositive claim. 
 \end{proof}


  \cref{thm:kop-1}  applies to all multi-agent systems. It immediately implies, for example, that 
  a necessary condition for the ATM to dispense cash is   $K_{atm}(\mbox{\tt good\_credit})$.
 The theorem is model independent; it does not depend on timing assumptions, on the topology of the system (even on whether agents communicate by message passing or via reading and writing to registers in a shared memory), or on the nature of the activity that is carried out. For every necessary condition for a conscious action,  {\em knowing} that the condition holds is also a necessary condition. 
  

\section{Coordinating Simultaneous Actions}
\label{sec:ck}

Recall that the language $\LKn$ contains formulas in which knowledge operators can be {\em nested} to arbitrary finite depth. 
It is sometimes useful to consider a state of knowledge called {\em common knowledge} that goes beyond any particular nested formula. 
Intuitively, a fact~$\psi$ is common knowledge if everyone knowing that everyone knows \ldots, that everyone knows the fact~$\psi$, to every finite depth. Common knowledge has a number of equivalent definitions, one of which is as follows: 

\begin{definition}[Common Knowledge]
\label{def:CK}
Fix a set of agents~$G$  and a fact~$\psi$. We
denote by~$C_G\psi$ 
the fact that 
 \defemph{$\bm{\psi}$ is common knowledge to~$\bm G$}. 
Its truth at points of a system~$R$ is defined by:
\vskip3mm
\begin{tabular}{lcl}
$(R,r,t)\sat C_G\psi$ & iff & $(R,r,t)\sat 
K_{i_1}K_{i_2}\cdots K_{i_m}\psi\!$ ~holds for all   $\langle i_1,i_2,\ldots,i_m\rangle\in G^m$ ~and all $m\ge 1$. 
\end{tabular}
\end{definition}

Common knowledge, a term coined by Lewis in \cite{Lew}, plays an important role in the analysis of games~\cite{Au}, distributed systems~\cite{HM1}, and many other multi-agent settings. 
Clearly, common knowledge is much stronger than ``plain''  knowledge. Indeed, $C_G\psi$ validly implies $K_j\psi$, for all  agents $j\in G$. 
Since common knowledge requires infinitely many facts to hold, 
it is not {\em a priori} obvious that $C_G\varphi$ can be attained at a reasonable cost, or even whether it can ever be attained at all, in settings of interest (see~\cite{ChM,FHMV,HM1}). 
We will now show that there are natural applications for which attaining common knowledge is essential. 

Intuitively, distinct actions are simultaneous in~$R$ if they can only  be performed together; whenever one is performed, all of them are performed simultaneously. It is possible to define simultaneous coordination formally in terms of necessary conditions: 
\begin{definition}[Simultaneous Actions]
Let~$G$ be a set of agents. We say that a set of actions $\bm{A}=\{\sfa_i\}_{i\in G}$ \defemph{is}
(necessarily) \defemph{simultaneous in}~$\bm{R}$ if $\ddo_i(\sfa_i)$ is a necessary condition for $\ddo_j(\sfa_j)$ in~$R$, for all $i,j\in G$. 
\end{definition}
Suppose that the actions in~$A$ are simultaneous in~$R$ in the above sense. Then the \KoP\   immediately implies (by Theorem~\ref{thm:kop-1}) that a necessary condition for performing an action in~$A$ 
is knowing that the other actions are also (currently) being performed. 
In fact, however, much more must be true. 
We now present a strong variant of the \KoP, which shows that in order to perform simultaneous actions agents must attain common knowledge of their necessary conditions. 
Notice that in order to allow a set of actions by the agents in~$G$ to be simultaneous, the system~$R$ must be sufficiently deterministic to ensure that if $i,j\in G$ are distinct agents and~$(R,r,t)\sat \ddo_i(\sfa)$ holds, then~$j$ will be scheduled to perform an action at $(r,t)$. For otherwise, there would be no way to ensure simultaneous execution of the actions by the agents in~$G$. Conscious actions fit this setting well in this case. We proceed as follows.

\begin{theorem}[C-K of Preconditions]
\label{thm:ckop}
Let $G$ be a set of agents and let $A=\{\sfa_i\}_{i\in G}$ be a set of necessarily simultaneous actions in the system~$R$. 
Moreover, suppose that each action~$\sfa_i\in A$ is a conscious action for  its agent~$i$ in~$R$. 
  If 
%
 $\psi$ is a necessary condition for  {$\ddo_i(\sfa_i)$} for some $i\in G$, then 
 $\,C_G\psi$ is a necessary condition for {$\ddo_j(\sfa_j)$}, for all $j\in G$.
\end{theorem}
\begin{proof}
Assume that~$A$ is a set of necessarily simultaneous actions for~$G$ in~$R$.
It is straightforward to show the following claim.
\begin{observation}
\label{obs1}
Let $\sfa_i,\sfa_j\in A$ be the actions for agents~$i$ and~$j$, respectively.  If a fact~$\varphi$ is a necessary condition for $\ddo_i(\sfa_i)$ in~$R$ then~$\varphi$ is also a necessary condition for $\ddo_j(\sfa_j)$ in~$R$. 
\end{observation}

To prove this observation notice that, 
by assumption, both (a) $R\sat \ddo_j(\sfa_j)\imp\ddo_i(\sfa_i)$ and 
(b)~$R\sat \ddo_i(\sfa_i)\imp\varphi$ hold. For all $(r,t)\in\Pts(R)$, if $(R,r,t)\sat\ddo_j(\sfa_j)$ then 
$(R,r,t)\sat\ddo_i(\sfa_i)$ holds by (a) and so \mbox{$(R,r,t)\sat\varphi$} by (b). Thus, 
$\varphi$ is a necessary condition for $\ddo_j(\sfa_j)$ in~$R$.


Assume that $\psi$ is a necessary condition for $\ddo_i(\sfa_i)$, for some $i\in G$.  
We shall prove by induction on~$m\ge 0$ that 
$K_{i_1}K_{i_2}\cdots K_{i_m}\psi$ is a necessary condition for $\ddo_j(\sfa_j)$ in~$R$, for every $j\in G$ and {\em all} 
sequences $\langle i_1,\ldots,i_m\rangle\in G^m$ (of~$m$ agent names from~$G$). 
This will establish that $(R,r,t)\sat\ddo_j(\sfa_j)$ implies $(R,r,t)\sat C_G\psi$ for all $(r,t)\in\Pts(R)$, and thus $C_G\psi$ is a necessary condition for {$\ddo_j(\sfa_j)$} for all $j\in G$, as claimed. 
\begin{itemize}
\item Base case: Let  $m=0$. The claim in this case is that if $\psi$ is a necessary condition for $\ddo_i(\sfa_i)$ then~$\psi$ is also a necessary condition for $\ddo_j(\sfa_j)$. This is precisely Observation~\ref{obs1}, with $\varphi\eqdef\psi$. 

\item Inductive step: Let $m\ge 1$, and assume that the claim holds for all $j'\in G$ and all sequences in~$G^{m-1}$. 
Fix $j\in G$ and a sequence $\langle i_1,i_2,\ldots,i_m\rangle\in G^m$. Its suffix $\langle i_2,\ldots,i_m\rangle$ is a sequence in~$G^{m-1}$. Thus, $K_{i_2}\cdots K_{i_m}\psi$ is a 
necessary condition for $\ddo_{i_1}(\sfa_{i_1})$ by the inductive hypothesis for~$m-1$ (applied to $G^{m-1}$ and agent $j'=i_1\in G$). Given that $\sfa_{i_1}$ is a conscious action by~$i_1$, we can apply  Theorem~\ref{thm:kop-1} to the necessary condition $K_{i_2}\cdots K_{i_m}\psi$ and obtain that $K_{i_1}K_{i_2}\cdots K_{i_m}\psi$ is a necessary condition for $\ddo_{i_1}(\sfa_{i_1})$. 
By Observation~\ref{obs1} we have that 
$K_{i_1}K_{i_2}\cdots K_{i_m}\psi$ is also a necessary condition for $\ddo_{j}(\sfa_j)$ in~$R$, and we are done.
\end{itemize}
\end{proof}

\subsection{Common Knowledge and the  Firing Squad 
Problem}
\label{sec:firing}
As an illustration of the applicability of Theorem~\ref{thm:ckop} to a concrete application, consider a simple version of the {\em Firing Squad} problem. In this instance, the set of agents~$G$ in the system must simultaneously perform an action (say each agent~$i\in G$ should perform the action~$\fire_i$) in response to the receipt, by any agent in~$G$, of a particular external input called a `$\go$' message. The $\fire_i$ action can stand for a simultaneous change in shared copies of a database, a public announcement at different sites of the system, or any other actions that need to take place simultaneously. Moreover, $\fire_i$ actions are allowed only if they are preceded by such a $\go$ message. For simplicity, we consider a case in which none of the agents in~$G$ may fail, and they all must satisfy the specification. 

Let $\recgo$ be a proposition that is true at $(r,t)\in\Pts(R)$ if a $\go$ message is received by any of the agents in~$G$ at a point $(r,t')$ of~$r$ 
at a time $t'\le t$. According to the specification of the Firing Squad problem, $\recgo$ is a necessary condition for the $\fire_i$ actions. An immediate consequence of Theorem~\ref{thm:ckop} is:
\begin{corollary}
\label{cor:fs}
$C_G\recgo$ is a necessary condition for all $\fire_i$ actions in the Firing Squad problem. 
\end{corollary}
Given Corollary~\ref{cor:fs}, any solution to the firing squad problem must first attain common knowledge that a $\go$ message has been received. 
It is well-known (see~\cite{Fagin1999revisited,HM1}) that common knowledge of a fact is observed simultaneously at all agents it involves. Suppose that  every 
$i\in G$ performs $\fire_i$ when $C_G\recgo$ first holds.  Since all agents in~$G$ will come to know that $C_G\recgo$ immediately,  they will fire simultaneously, as required by the problem specification. Indeed, Theorem~\ref{thm:ckop} shows that this is the first time at which they {\em can} perform according to a correct protocol. Implementing simultaneous tasks such as the Firing Squad therefore inherently involves, and often reduces to, ensuring and detecting $C_G\recgo$.
Recall that depending on the properties of the system, attaining such common knowledge might be impossible in some cases, or it might incur a substantial cost in others. Just as in the case of the \KoP, this necessity is not due to our formalism. It is only exposed by our analysis. In every protocol that implements such a task correctly, the firing actions cannot be performed unless  $C_G\recgo$ is attained.

There is an extensive literature on using common knowledge to obtain optimal protocols for simultaneous tasks \cite{DHM,DM,HerMT,MizMos,MizMos1,MT, Neig90,neiger1999using,NeigTuttle}. Typically, they involve an explicit proof that common knowledge of a particular fact is a necessary condition for performing a set~$A$ of necessarily simultaneous actions.  \cref{thm:ckop} or a variant of it suited for fault-tolerant systems can be used to establish this result in all of these cases. 
Moreover, one of the main insights from the analysis of \cite{HM1} and of \cite{Fagin1999revisited} is that when simultaneous actions are performed, the participating agents have common knowledge that they are being performed. \cref{thm:ckop} is a strict generalization of this fact.

\section{Temporally Ordering Actions}
\label{sec:nested}
So far, we have seen two essential connections between knowledge and  coordinated action: performing actions requires knowledge of their necessary conditions, and performing simultaneous actions requires common knowledge of their necessary conditions. We now further extend the connection between states of knowledge and coordination, by showing that temporally ordering actions depends on attaining nested knowledge of necessary conditions.
Following~\cite{BZMacm}, we define temporally ordered actions:

\begin{definition}[Ben Zvi and Moses]
\label{def:ordered}
A sequence of actions $\bm{\langle \sfa_1,\ldots,\sfa_k\rangle}$ (for agents ~$1,\ldots,k$, respectively) 
\defemph{is} (linearly) \defemph{ordered in~$\bm{R}$} if $\,\dDid_{j-1}(\sfa_{j-1})$ is a necessary condition for $\ddo_j(\sfa_j)$ in~$R$. 
\end{definition}
Observe that this definition does not force an action~$\sfa_j$ to occur in a run in which~$\sfa_{j-1}$ occurs. 
Rather, if an action~$\sfa_j$ is performed in a given run, then it must be preceded by all actions~$\sfa_1,\ldots,\sfa_{j-1}$. 
Moreover, if we denote the time at which an action~$\sfa_i$ is performed in a run~$r$ by~$t_i$, then we require that $t_{j-1}\le t_j$ for every action~$\sfa_j$  performed in~$r$. 
\begin{claim}
\label{claim1}
Assume that the sequence $\langle \sfa_1,\ldots,\sfa_k\rangle$ is ordered in~$R$. 
Then $R\sat \big(\dDid_j(\sfa_j)\imp \dDid_{j-1}(\sfa_{j-1})\big)$ for all 
$2\le j\le k$. 
\end{claim}

\begin{proof}
Assume that $(R,r,t)\sat\dDid_j(\sfa_j)$. Then, by  definition of~$\dDid_j(\sfa_j)$,   we have $(R,r,\hat{t}\,)\sat \ddo_j(\sfa_j)$ for some $\hat{t}\le t$. 
The fact that $\langle \sfa_1,\ldots,\sfa_k\rangle$ is ordered in~$R$ implies that $\dDid_{j-1}(\sfa_{j-1})$ is a necessary condition for
$\ddo_j(\sfa_j)$ in the system~$R$, and so $(R,r,\hat{t}\,)\sat \dDid_{j-1}(\sfa_{j-1})$. 
Since $\dDid_{j-1}(\sfa_{j-1})$ is a stable fact  and $t\ge\hat{t}$, we obtain that  
$(R,r,t)\sat \dDid_{j-1}(\sfa_{j-1})$. The claim follows.
\end{proof}

We say that a fact $\varphi$ is \defemph{stable in~$\bm{R}$} if once true, $\varphi$
remains true. Formally,  if ~$(R,r,t)\sat\varphi$ and $t'>t$
then \mbox{$(R,r,t')\sat \varphi$}, for all $r\in R$ and $t,t'\ge 0$.
Notice that while $\ddo_i(\sfa)$ is, in general, not a stable fact, $\dDid_i(\sfa)$ is always stable. 

\begin{definition} 
\label{def:selective}
We say that \defemph{agent~$\bm{i}$ recalls $\bm{\psi}$ in}~$\bm{R}$ if the fact $K_i\psi$ is stable in~$R$. 
\end{definition}
The notion of {\em perfect recall}, capturing the assumption that agents remember all events that they take part in, is popular in the analysis of games and multi-agent systems~%
\cite{FHMV,Selten75}. 
While perfect recall is a nontrivial assumption often requiring significant storage costs, selective recall of single facts such as $\ddo_j(\sfa_j)$ is a much weaker assumption, that can be assumed of a system~$R$ essentially  without loss of generality. By adding a single bit to Agent~$j$'s local state, whose value is 0 as long as~$j$ has not performed $\sfa_j$ and~1 once the action has been performed, we can obtain a system~$R'$ that is isomorphic to~$R$, in which 
Agent~$j$ recalls $\ddo_j(\sfa_j)$. 

\begin{claim}
\label{claim2}
Assume that $\sfa_j$ is a conscious action for~$j$ in~$R$, 
and that $j$ recalls $\dDid_j(\sfa_j)$ in~$R$. Then $\dDid_j(\sfa_j)$
is a $j$-local fact in~$R$.
\end{claim}

\begin{proof}
Suppose that $(R,r,t)\sat \dDid_j(\sfa_j)$. Then, by definition of $\dDid_j(\sfa_j)$, we have  $(R,r,\hat{t}\,)\sat \ddo_j(\sfa_j)$ for some 
$\hat{t}\le t$.
Choose an arbitrary  $(r',t')\in \Pts(R)$ satisfying that
$(r',t')\rlsim{j}(r,\hat{t}\,)$. It follows that $(R,r',t')\sat\ddo_j(\sfa_j)$ since
$\sfa_j$ is a conscious action for~$j$ in~$R$. 
By definition of~$\dDid_j(\sfa_j)$ it follows that
$(R,r',t')\sat\dDid_j(\sfa_j)$.
Now, by definition of $\sat$ for $K_j$  we have that $(R,r,\hat{t}\,)\sat
K_j\dDid_j(\sfa_j)$. By assumption, $j$ recalls $\dDid_j(\sfa_j)$ in~$R$, and so $K_j\dDid_j(\sfa_j)$ is stable in~$R$. Thus, since $t\ge\hat{t}$, we obtain 
that $(R,r,t)\sat K_j\dDid_j(\sfa_j)$, as claimed. 
\end{proof}

We can now show: 

\begin{theorem}[Ordering and Nested Knowledge]
\label{thm:nested}
\label{thm:nkop}
 Assume that 
\begin{itemize}
\item 
the actions $\langle \sfa_1,\ldots,\sfa_k\rangle$ are \emph{ordered} in~$R$,
\item each agent~$j=1,\ldots,k$ recalls $\,\dDid_j(\sfa_j)$ in~$R$, 
\item $\sfa_j$ is a conscious 
action for~$j$ in~$R$, for all $j=1,\ldots,k$, and 
\item $\psi$ is a $\,$\underline{\emph{stable}}$\,$ 
necessary condition for the first action~$\ddo_1(\sfa_1)$ in~$R$ 
\end{itemize}
Then  $\,K_jK_{j-1}\cdots K_1\psi\,$ is a necessary condition for 
the~$j^{\,\mathrm{th}}$ action $\ddo_j(\sfa_j)$ in~$R$,  for all $j\le k$.
\end{theorem}

\begin{proof}
Assuming the conditions of the theorem, we will prove
by induction on~$j\le k$ that 
$\dDid_j(\sfa_j)$ validly implies $K_jK_{j-1}\cdots K_1\psi$ in~$R$. 
Since $\ddo_j(\sfa_j)$ validly implies $\dDid_j(\sfa_j)$ by definition of $\dDid_j(\sfa_j)$, this will yield that 
$K_jK_{j-1}\cdots K_1\psi$ is a necessary condition for $\ddo_j(\sfa_j)$ in~$R$, 
as claimed. 
We proceed with the inductive argument.

\begin{itemize}
\item Base case $j=1$:\quad
Assume that $(R,r,t)\sat \dDid_1(\sfa_1)$. Claim~\ref{claim2}  implies
that $(R,r,t)\sat K_1\dDid_1(\sfa_1)$. Let  $(r',t')\in \Pts(R)$ be an arbitrary point satisfying  that 
$(r',t')\rlsim{1}(r,t)$. Then $(R,r',t')\sat \dDid_1(\sfa_1)$ by the knowledge property.
Thus, $(R,r',\hat{t}\,)\sat \ddo_1(\sfa_1)$ holds for some $\hat{t}\le t'$, and because $\psi$ is a necessary condition
for~$\ddo_1(\sfa_1)$ in~$R$, we obtain that $(R,r,\hat{t}\,)\sat\psi$. Since~$\psi$ is
stable and $t'\ge \hat{t}$, we have that \mbox{$(R,r',t')\sat\psi$}. By
choice of $(r',t')$ we have that $(R,r,t)\sat K_1\psi$, as claimed. 

\item Inductive step: Let $j>1$ and assume that 
$K_{j-1}\cdots K_1\psi$ is a necessary condition  for $\dDid_{j-1}(\sfa_{j-1})$ in~$R$. Moreover, let  $(R,r,t)\sat\dDid_j(\sfa_j)$. 
Since  $\sfa_j$ is a conscious action for~$j$, Claim~\ref{claim2} implies  that $(R,r,t)\sat K_j\dDid_j(\sfa_j)$.
Choose an arbitrary  $(r',t')\in \Pts(R)$ satisfying that 
$(r',t')\rlsim{j}(r,t)$. 
By definition of~$K_j$, it follows that $(R,r',t')\sat\dDid_j(\sfa_j)$. 
By Claim~\ref{claim1},  since the sequence $\langle \sfa_1,\ldots,\sfa_k\rangle$ is ordered in~$R$ and $j>1$ we have that $(R,r',t')\sat\dDid_{j-1}(\sfa_{j-1})$.
We now apply the inductive hypothesis to obtain that 
$(R,r',t')\sat K_{j-1}\cdots K_1\psi$. 
Finally, we obtain 
that $(R,r,t)\sat K_jK_{j-1}\cdots K_1\psi$ 
by choice of~$(r',t')$ and the definition of `$\sat$'  for~$K_j$.
The claim now follows.
\end{itemize}
\end{proof}

A slightly more restricted version of Theorem~\ref{thm:nested} was proved in~\cite{BZMacm}. Rather than consider an arbitrary necessary condition for $\alpha_1$, they proved a version for the case in which the first action~$\sfa_1$ is triggered by an external input to agent~1. Technically, the proofs are quite similar. 

Theorem~\ref{thm:nested} provides a necessary, but possibly not sufficient, condition for 
ordering actions in distributed systems. 
If agent~$j$ acts {strictly later} than when  $K_jK_{j-1}\cdots K_1\psi$ first holds, then it may be inappropriate for agent~$j+1$ to act when it knows 
that the fact $K_jK_{j-1}\cdots K_1\psi$ holds (i.e., when $K_{j+1}K_j\cdots K_1\psi$ first holds). Nevertheless, Theorem~\ref{thm:nested} is often very useful because it can be used as a guide for efficiently, and sometimes even optimally, performing a sequence of ordered actions. 
Intuitively, suppose that we have a protocol whose goal is to perform $\langle \sfa_1,\ldots,\sfa_k\rangle$ in response to an externally generated  trigger~$\psi$ (such as the `$\go$' message in Firing Squad). In particular, assume that~$\psi$ is a necessary condition for~$\sfa_1$. Keeping the communication aspects of this protocol fixed, an optimally fast solution would be for each agent~$j\le k$ to perform $\sfa_j$ when $K_jK_{j-1}\cdots K_1\psi$ first holds.
Let $R$ be the set of runs of such a protocol with $r\in~R$, and let 
$t_j$ and $t_{j-1}$ be the earliest times at which $(R,r,t_j)\sat K_jK_{j-1}\cdots K_1\psi$ and $(R,r,t_{j-1})\sat K_{j-1}\cdots K_1\psi$ hold in a run~$r$, respectively. 
The knowledge property guarantees that $K_jK_{j-1}\cdots K_1\psi$ validly implies that $K_{j-1}\cdots K_1\psi$ in~$R$, and so $t_j\ge t_{j-1}$. 
Since, by assumption, $\sfa_j$ is performed at time~$t_j$ and $\sfa_{j-1}$ at $t_{j-1}$, we have that agents perform actions in linear temporal order, as required by Definition~\ref{def:ordered}. Clearly, none of the actions can be performed any earlier, as Theorem~\ref{thm:nested} shows. We conclude that in time-efficient protocols, the nested knowledge formula presented by the theorem can be  both necessary and sufficient. In this sense, Theorem~\ref{thm:nested} suggests a recipe for obtaining time-efficient solutions for ordering actions.

Just as \cref{thm:ckop} implies that common knowledge is a necessary condition for simultaneous actions, we now have by \cref{thm:nkop} that nested knowledge is a necessary condition for performing actions in linear temporal order. And just as there is an established literature on when common knowledge is and is not attainable and on how it may arise, there are results concerning the communication structure that underlies attaining nested knowledge. Indeed, in a seminal paper~\cite{ChM}, Chandy and Misra showed that in asynchronous systems~$R$, if $(R,r,t)\sat\neg\varphi$ 
and at a time \mbox{$t'>t$}
$(R,r,t')\sat K_jK_{j-1}\cdots K_1\varphi$, 
then there must be a message chain in the run~$r$ between times~$t$ and~$t'$, passing through the agents~1,2,\ldots,$j$ in this order (possibly involving additional agents as well). Given \cref{thm:nkop}, this implies that the only way to coordinate actions in a linear temporal order in an asynchronous setting is by way of such message chains.%
\footnote{\cref{thm:kop-1,thm:nkop} depend on conscious actions and therefore do not apply to asynchronous systems. Nevertheless, variants of these theorems can be presented that do apply to asynchronous systems and nondeterministic protocols. Details will appear in~\cite{Mono}.}

More recently, Ben Zvi and Moses extended Chandy and Misra's work to systems in which communication is not asynchronous, but rather agents may have access to clocks and the transmission time for each of the channels is bounded \cite{BZMacm}. They show that a communication structure called a {\em centipede} must be constructed in order to obtain nested knowledge of spontaneous facts such as the arrival of an external input. 
They prove a slightly more restricted instance of \cref{thm:nkop} (without using \KoP\ directly), and use it to show that ordering actions in their setting requires the construction of the appropriate centipedes. 
Finally, Parikh and Krasucki analyze the ability to create {\em levels of knowledge} consisting of collections of nested knowledge formulas in~\cite{PK92}. \cref{thm:nkop} relates levels of knowledge to coordination.

\section{Discussion}
\label{sec:discussion}
This paper formulated the knowledge of preconditions principle and presented three theorems relating knowledge and coordinated action: the first is the \KoP\  itself---necessary conditions for an action must be {\em known} to hold when the action is performed. Next, we showed that necessary conditions for simultaneous actions must be commonly known when the actions are taken. Finally, nested knowledge is a necessary condition for coordinating linearly ordered actions. 
The latter two are fairly direct consequences of the \KoP.  We discussed some of the uses of the latter two results in Sections~\ref{sec:ck} and~\ref{sec:nested}. Indeed the \KoP\ has many further implications. 

In recent years, several works that make use of \KoP\ have appeared, citing the 
unpublished~\cite{Mono}. For example, Casta\~{n}eda, Gonczarowski and Moses used the \KoP\ to analyze the consensus problem~\cite{CGM}, in which agents need to agree on a binary value in a fault-prone system. They designed a protocol in two steps---applying the \KoP\ once to derive a rule by which, roughly, agents decide on~0 when they know of an initial value of~0. Then, they applied the \KoP\ again assuming that this is the rule used for making decisions on~0, and obtained a rule involving nested knowledge (roughly, a statement of the form  ``knowing that nobody knows~0'')  for deciding on the value~1. The result of their analysis was a very efficient solution to consensus that is optimal in a strong sense: It is the first unbeatable consensus protocol. No protocol can strictly dominate it, by having processes always decide at least as fast, and sometimes strictly faster, than this protocol does. The work of~\cite{CGM} complements an earlier work by Halpern, Moses and Waarts~\cite{HalMoWa2001}, in which a fixed point analysis of optimal consensus was obtained. The latter, too, is closely related to the \KoP.
 
 Gonczarowski and Moses used the \KoP\ to analyze the epistemic requirements of more general forms of coordination~\cite{GoM}. Namely, they considered a setting in which~$k$ agents need to perform actions, and there are time bounds on the relative times at which the actions of any pair of agents is performed. 
The simple instance in which all bounds are~0 is precisely that of the  simultaneous actions considered in \cref{sec:ck}. They show that such coordination requires vectorial fixed points of knowledge conditions, which are naturally related to fixed points and equilibria. 
The papers \cite{BzMICLA2013,BzMTARK2013,BZMacm,GoM} together can all be viewed as making use of the \KoP\ to provide insights into the interaction between time and communication for coordinating actions in a distributed and multi-agent system. Describing them is beyond the scope of the current paper. 

The most significant aspect of the \KoP, in our view, is the fact that it places a new emphasis on the epistemic aspects of problem solving in a multi-agent system. Simple necessary conditions induce epistemic conditions. Thus, in order to act correctly, one needs a mechanism ensuring that the agents obtain the necessary knowledge, and that they discover that they have this knowledge. Most problems and solutions are not posed or described in this fashion. We believe that the \KoP\ encapsulates an important connection between knowledge, action and coordination that will find many applications in the future. 
\bibliographystyle{eptcs}
\bibliography{newtkop,newz1}

\begin{thebibliography}{10}
\providecommand{\bibitemdeclare}[2]{}
\providecommand{\surnamestart}{}
\providecommand{\surnameend}{}
\providecommand{\urlprefix}{Available at }
\providecommand{\url}[1]{\texttt{#1}}
\providecommand{\href}[2]{\texttt{#2}}
\providecommand{\urlalt}[2]{\href{#1}{#2}}
\providecommand{\doi}[1]{doi:\urlalt{http://dx.doi.org/#1}{#1}}
\providecommand{\bibinfo}[2]{#2}

\bibitemdeclare{book}{AWbook}
\bibitem{AWbook}
\bibinfo{author}{Hagit \surnamestart Attiya\surnameend} \&
  \bibinfo{author}{Jennifer \surnamestart Welch\surnameend}
  (\bibinfo{year}{2004}): \emph{\bibinfo{title}{Distributed Computing:
  Fundamentals, Simulations and Advanced Topics}}.
\newblock \bibinfo{publisher}{John Wiley \& Sons}, \doi{10.1002/0471478210}.

\bibitemdeclare{article}{Au}
\bibitem{Au}
\bibinfo{author}{R.~J. \surnamestart Aumann\surnameend} (\bibinfo{year}{1976}):
  \emph{\bibinfo{title}{Agreeing to disagree}}.
\newblock {\sl \bibinfo{journal}{Annals of Statistics}}
  \bibinfo{volume}{4}(\bibinfo{number}{6}), pp. \bibinfo{pages}{1236--1239},
  \doi{10.1214/aos/1176343654}.

\bibitemdeclare{inproceedings}{BzMICLA2013}
\bibitem{BzMICLA2013}
\bibinfo{author}{Ido \surnamestart Ben{-}Zvi\surnameend} \&
  \bibinfo{author}{Yoram \surnamestart Moses\surnameend}
  (\bibinfo{year}{2013}): \emph{\bibinfo{title}{Agent-Time Epistemics and
  Coordination}}.
\newblock In: {\sl \bibinfo{booktitle}{Proceedings of ICLA}}, pp.
  \bibinfo{pages}{97--108}, \doi{10.1007/978-3-642-36039-8\_9}.

\bibitemdeclare{inproceedings}{BzMTARK2013}
\bibitem{BzMTARK2013}
\bibinfo{author}{Ido \surnamestart Ben{-}Zvi\surnameend} \&
  \bibinfo{author}{Yoram \surnamestart Moses\surnameend}
  (\bibinfo{year}{2013}): \emph{\bibinfo{title}{The Shape of Reactive
  Coordination Tasks}}.
\newblock In: {\sl \bibinfo{booktitle}{Proceedings of TARK}},
  \bibinfo{series}{TARK XIV}, pp. \bibinfo{pages}{29--38}.

\bibitemdeclare{article}{BZMacm}
\bibitem{BZMacm}
\bibinfo{author}{Ido \surnamestart Ben{-}Zvi\surnameend} \&
  \bibinfo{author}{Yoram \surnamestart Moses\surnameend}
  (\bibinfo{year}{2014}): \emph{\bibinfo{title}{Beyond Lamport's
  \emph{Happened-before}: On Time Bounds and the Ordering of Events in
  Distributed Systems}}.
\newblock {\sl \bibinfo{journal}{J. {ACM}}}
  \bibinfo{volume}{61}(\bibinfo{number}{2}), p.~\bibinfo{pages}{13},
  \doi{10.1145/2542181}.

\bibitemdeclare{inproceedings}{CGM}
\bibitem{CGM}
\bibinfo{author}{Armando \surnamestart Casta{\~n}eda\surnameend},
  \bibinfo{author}{Yannai~A \surnamestart Gonczarowski\surnameend} \&
  \bibinfo{author}{Yoram \surnamestart Moses\surnameend}
  (\bibinfo{year}{2014}): \emph{\bibinfo{title}{Unbeatable Consensus}}.
\newblock In: {\sl \bibinfo{booktitle}{Proceedings of DISC}},
  \bibinfo{publisher}{Springer}, pp. \bibinfo{pages}{91--106},
  \doi{10.1007/978-3-662-45174-8\_7}.

\bibitemdeclare{article}{ChM}
\bibitem{ChM}
\bibinfo{author}{K.~M. \surnamestart Chandy\surnameend} \&
  \bibinfo{author}{J.~\surnamestart Misra\surnameend} (\bibinfo{year}{1986}):
  \emph{\bibinfo{title}{How processes learn}}.
\newblock {\sl \bibinfo{journal}{Distributed Computing}}
  \bibinfo{volume}{1}(\bibinfo{number}{1}), pp. \bibinfo{pages}{40--52},
  \doi{10.1007/BF01843569}.

\bibitemdeclare{article}{DHM}
\bibitem{DHM}
\bibinfo{author}{Danny \surnamestart Dolev\surnameend}, \bibinfo{author}{Ezra~N
  \surnamestart Hoch\surnameend} \& \bibinfo{author}{Yoram \surnamestart
  Moses\surnameend} (\bibinfo{year}{2012}): \emph{\bibinfo{title}{An optimal
  self-stabilizing firing squad}}.
\newblock {\sl \bibinfo{journal}{SIAM Journal on Computing}}
  \bibinfo{volume}{41}(\bibinfo{number}{2}), pp. \bibinfo{pages}{415--435},
  \doi{10.1137/090776512}.

\bibitemdeclare{article}{DM}
\bibitem{DM}
\bibinfo{author}{C.~\surnamestart Dwork\surnameend} \&
  \bibinfo{author}{Y.~\surnamestart Moses\surnameend} (\bibinfo{year}{1990}):
  \emph{\bibinfo{title}{Knowledge and common knowledge in a {B}yzantine
  environment: crash failures}}.
\newblock {\sl \bibinfo{journal}{Information and Computation}}
  \bibinfo{volume}{88}(\bibinfo{number}{2}), pp. \bibinfo{pages}{156--186},
  \doi{10.1016/0890-5401(90)90014-9}.

\bibitemdeclare{book}{FHMV}
\bibitem{FHMV}
\bibinfo{author}{R.~\surnamestart Fagin\surnameend}, \bibinfo{author}{J.~Y.
  \surnamestart Halpern\surnameend}, \bibinfo{author}{Y.~\surnamestart
  Moses\surnameend} \& \bibinfo{author}{M.~Y. \surnamestart Vardi\surnameend}
  (\bibinfo{year}{2003}): \emph{\bibinfo{title}{Reasoning about Knowledge}}.
\newblock \bibinfo{publisher}{MIT Press}, \bibinfo{address}{Cambridge, Mass.}

\bibitemdeclare{article}{Fagin1999revisited}
\bibitem{Fagin1999revisited}
\bibinfo{author}{Ronald \surnamestart Fagin\surnameend},
  \bibinfo{author}{Joseph~Y. \surnamestart Halpern\surnameend},
  \bibinfo{author}{Yoram \surnamestart Moses\surnameend} \&
  \bibinfo{author}{Vardi.~Moshe \surnamestart Y.\surnameend}
  (\bibinfo{year}{1999}): \emph{\bibinfo{title}{Common knowledge revisited}}.
\newblock {\sl \bibinfo{journal}{Annals of Pure and Applied Logic}}
  \bibinfo{volume}{96}(\bibinfo{number}{1–3}),
  \doi{10.1016/S0168-0072(98)00033-5}.

\bibitemdeclare{inproceedings}{GoM}
\bibitem{GoM}
\bibinfo{author}{Y~\surnamestart Gonczarowski\surnameend} \&
  \bibinfo{author}{Y~\surnamestart Moses\surnameend} (\bibinfo{year}{2013}):
  \emph{\bibinfo{title}{Timely common knowledge: Characterising asymmetric
  distributed coordination via vectorial fixed points}}.
\newblock In: {\sl \bibinfo{booktitle}{Proceedings of TARK XIV}}.

\bibitemdeclare{article}{HM1}
\bibitem{HM1}
\bibinfo{author}{J.~Y. \surnamestart Halpern\surnameend} \&
  \bibinfo{author}{Y.~\surnamestart Moses\surnameend} (\bibinfo{year}{1990}):
  \emph{\bibinfo{title}{Knowledge and Common Knowledge in a Distributed
  Environment}}.
\newblock {\sl \bibinfo{journal}{Journal of the ACM}}
  \bibinfo{volume}{37}(\bibinfo{number}{3}), pp. \bibinfo{pages}{549--587},
  \doi{10.1145/800222.806735}.
\newblock \bibinfo{note}{A preliminary version appeared in {\em Proc.~3rd ACM
  PODC}, 1984}.

\bibitemdeclare{article}{HM2}
\bibitem{HM2}
\bibinfo{author}{J.~Y. \surnamestart Halpern\surnameend} \&
  \bibinfo{author}{Y.~\surnamestart Moses\surnameend} (\bibinfo{year}{1992}):
  \emph{\bibinfo{title}{A guide to completeness and complexity for modal logics
  of knowledge and belief}}.
\newblock {\sl \bibinfo{journal}{Artificial Intelligence}}
  \bibinfo{volume}{54}, pp. \bibinfo{pages}{319--379},
  \doi{10.1016/0004-3702(92)90049-4}.

\bibitemdeclare{article}{HalMoWa2001}
\bibitem{HalMoWa2001}
\bibinfo{author}{Joseph~Y. \surnamestart Halpern\surnameend},
  \bibinfo{author}{Yoram \surnamestart Moses\surnameend} \&
  \bibinfo{author}{Orli \surnamestart Waarts\surnameend}
  (\bibinfo{year}{2001}): \emph{\bibinfo{title}{A Characterization of Eventual
  Byzantine Agreement}}.
\newblock {\sl \bibinfo{journal}{SIAM J. Comput.}}
  \bibinfo{volume}{31}(\bibinfo{number}{3}), pp. \bibinfo{pages}{838--865},
  \doi{10.1137/S0097539798340217}.

\bibitemdeclare{inproceedings}{HerMT}
\bibitem{HerMT}
\bibinfo{author}{Maurice~P \surnamestart Herlihy\surnameend},
  \bibinfo{author}{Yoram \surnamestart Moses\surnameend} \&
  \bibinfo{author}{Mark~R \surnamestart Tuttle\surnameend}
  (\bibinfo{year}{2011}): \emph{\bibinfo{title}{Transforming worst-case optimal
  solutions for simultaneous tasks into all-case optimal solutions}}.
\newblock In: {\sl \bibinfo{booktitle}{Proceedings of the 30th annual ACM
  SIGACT-SIGOPS symposium on Principles of distributed computing}},
  \bibinfo{organization}{ACM}, pp. \bibinfo{pages}{231--238},
  \doi{10.1145/1993806.1993849}.

\bibitemdeclare{inproceedings}{lelann}
\bibitem{lelann}
\bibinfo{author}{G{\'e}rard \surnamestart Le~Lann\surnameend}
  (\bibinfo{year}{1977}): \emph{\bibinfo{title}{Distributed Systems-Towards a
  Formal Approach}}.
\newblock In: {\sl \bibinfo{booktitle}{IFIP Congress}}, \bibinfo{volume}{7},
  \bibinfo{organization}{Toronto}, pp. \bibinfo{pages}{155--160}.

\bibitemdeclare{book}{Lew}
\bibitem{Lew}
\bibinfo{author}{D.~\surnamestart Lewis\surnameend} (\bibinfo{year}{1969}):
  \emph{\bibinfo{title}{Convention, A Philosophical Study}}.
\newblock \bibinfo{publisher}{Harvard University Press},
  \bibinfo{address}{Cambridge, Mass.}

\bibitemdeclare{incollection}{MH}
\bibitem{MH}
\bibinfo{author}{J.~\surnamestart McCarthy\surnameend} \&
  \bibinfo{author}{P.~J. \surnamestart Hayes\surnameend}
  (\bibinfo{year}{1969}): \emph{\bibinfo{title}{Some Philosophical Problems
  From the Standpoint of Artificial Intelligence}}.
\newblock In: {\sl \bibinfo{booktitle}{Machine Intelligence 4}},
  \bibinfo{publisher}{Edinburgh University Press}, pp.
  \bibinfo{pages}{463--502}, \doi{10.1016/b978-0-934613-03-3.50033-7}.

\bibitemdeclare{article}{MizMos}
\bibitem{MizMos}
\bibinfo{author}{Tal \surnamestart Mizrahi\surnameend} \&
  \bibinfo{author}{Yoram \surnamestart Moses\surnameend}
  (\bibinfo{year}{2008}): \emph{\bibinfo{title}{Continuous consensus via common
  knowledge}}.
\newblock {\sl \bibinfo{journal}{Distributed Computing}}
  \bibinfo{volume}{20}(\bibinfo{number}{5}), pp. \bibinfo{pages}{305--321},
  \doi{10.1007/s00446-007-0049-6}.

\bibitemdeclare{incollection}{MizMos1}
\bibitem{MizMos1}
\bibinfo{author}{Tal \surnamestart Mizrahi\surnameend} \&
  \bibinfo{author}{Yoram \surnamestart Moses\surnameend}
  (\bibinfo{year}{2008}): \emph{\bibinfo{title}{Continuous consensus with
  ambiguous failures}}.
\newblock In: {\sl \bibinfo{booktitle}{Distributed Computing and Networking}},
  \bibinfo{publisher}{Springer}, pp. \bibinfo{pages}{73--85},
  \doi{10.1016/j.tcs.2010.04.025}.

\bibitemdeclare{unpublished}{Mono}
\bibitem{Mono}
\bibinfo{author}{Y~\surnamestart Moses\surnameend} (\bibinfo{year}{2016}):
  \bibinfo{note}{{\it Knowledge and Coordinated Action}, to appear}.

\bibitemdeclare{article}{MT}
\bibitem{MT}
\bibinfo{author}{Y.~\surnamestart Moses\surnameend} \& \bibinfo{author}{M.~R.
  \surnamestart Tuttle\surnameend} (\bibinfo{year}{1988}):
  \emph{\bibinfo{title}{Programming simultaneous actions using common
  knowledge}}.
\newblock {\sl \bibinfo{journal}{Algorithmica}} \bibinfo{volume}{3}, pp.
  \bibinfo{pages}{121--169}, \doi{10.1007/BF01762112}.

\bibitemdeclare{unpublished}{Neig90}
\bibitem{Neig90}
\bibinfo{author}{G.~\surnamestart Neiger\surnameend} (\bibinfo{year}{1990}):
  \emph{\bibinfo{title}{Consistent coordination and continual common knowledge.
  {M}anuscript}}.

\bibitemdeclare{article}{NeigTuttle}
\bibitem{NeigTuttle}
\bibinfo{author}{G.~\surnamestart Neiger\surnameend} \& \bibinfo{author}{M.~R.
  \surnamestart Tuttle\surnameend} (\bibinfo{year}{1993}):
  \emph{\bibinfo{title}{Common knowledge and consistent simultaneous
  coordination}}.
\newblock {\sl \bibinfo{journal}{Distributed Computing}}
  \bibinfo{volume}{6}(\bibinfo{number}{3}), pp. \bibinfo{pages}{334--352},
  \doi{10.1007/BF02242706}.

\bibitemdeclare{article}{neiger1999using}
\bibitem{neiger1999using}
\bibinfo{author}{Gil \surnamestart Neiger\surnameend} \&
  \bibinfo{author}{Rida~A \surnamestart Bazzi\surnameend}
  (\bibinfo{year}{1999}): \emph{\bibinfo{title}{Using knowledge to optimally
  achieve coordination in distributed systems}}.
\newblock {\sl \bibinfo{journal}{Theoretical computer science}}
  \bibinfo{volume}{220}(\bibinfo{number}{1}), pp. \bibinfo{pages}{31--65},
  \doi{10.1016/S0304-3975(98)00236-9}.

\bibitemdeclare{article}{PK92}
\bibitem{PK92}
\bibinfo{author}{R.~\surnamestart Parikh\surnameend} \&
  \bibinfo{author}{P.~\surnamestart Krasucki\surnameend}
  (\bibinfo{year}{1992}): \emph{\bibinfo{title}{Levels of knowledge in
  distributed computing}}.
\newblock {\sl \bibinfo{journal}{S\={a}dhan\={a}}}
  \bibinfo{volume}{17}(\bibinfo{number}{1}), pp. \bibinfo{pages}{167--191},
  \doi{10.1007/bf02811342}.

\bibitemdeclare{article}{schneider1990}
\bibitem{schneider1990}
\bibinfo{author}{Fred~B \surnamestart Schneider\surnameend}
  (\bibinfo{year}{1990}): \emph{\bibinfo{title}{Implementing fault-tolerant
  services using the state machine approach: A tutorial}}.
\newblock {\sl \bibinfo{journal}{ACM Computing Surveys (CSUR)}}
  \bibinfo{volume}{22}(\bibinfo{number}{4}), pp. \bibinfo{pages}{299--319},
  \doi{10.1145/98163.98167}.

\bibitemdeclare{article}{Selten75}
\bibitem{Selten75}
\bibinfo{author}{R.~\surnamestart Selten\surnameend} (\bibinfo{year}{1975}):
  \emph{\bibinfo{title}{Reexamination of the perfectness concept for
  equilibrium points in extensive games}}.
\newblock {\sl \bibinfo{journal}{International Journal of Game Theory}}
  \bibinfo{volume}{4}, pp. \bibinfo{pages}{25--55}, \doi{10.1145/2542181}.

\end{thebibliography}
\end{document}